\documentclass[conference,10pt]{IEEEtran}
\long\def\@makecaption#1#2{}
\usepackage{graphicx}
\usepackage{subcaption}
\usepackage{paralist}
\usepackage[utf8]{inputenc}
\usepackage{amsmath}
\usepackage{mathtools}
\usepackage{amssymb}
\usepackage[standard,amsmath,thmmarks]{ntheorem}
\usepackage{proof} 
\usepackage{prftree}
\usepackage{mathpartir}
\usepackage[shortlabels]{enumitem}
\usepackage{breqn}
\usepackage{mathtools}
\usepackage{color}
\usepackage{listings}
\usepackage{lstautogobble}  
\usepackage{color}          
\usepackage{zi4}            
\usepackage{tikz-cd}
\definecolor{bluekeywords}{rgb}{0.13, 0.13, 1}
\definecolor{greencomments}{rgb}{0, 0.5, 0}
\definecolor{redstrings}{rgb}{0.9, 0, 0}
\definecolor{graynumbers}{rgb}{0.5, 0.5, 0.5}

\usepackage{listings}
\lstdefinelanguage{upl}%
{morekeywords={%
		assume,assert,%
		while,if
	},%
	numbers=none, %
	mathescape=true,%
	sensitive,%
	commentstyle=\color{blue},
	morecomment=[l]//,%
	morecomment=[s]{/*}{*/},%
	morestring=[b]",%
	morestring=[b]',%
	showstringspaces=false
}[keywords,comments,strings]%
\usepackage{cite} 
\usepackage{xspace}
\usepackage{cleveref}

\newcommand{\terms}{\mathcal{T}}
\newcommand{\const}{\mathcal{C}}
\newcommand{\funs}{\mathcal{F}}
\newcommand{\consts}{\const}
\newcommand{\rename}{\rightarrowtail}

\newcommand{\final}{\bot}

\newcommand{\euf}{EUF\xspace}

\newcommand{\pv}[1]{\texttt{#1}}
\newcommand{\hg}[1]{\textcolor{blue}{[HG: #1]}}
\newcommand{\sharon}[1]{\textcolor{red}{[sharon: #1]}}
\newcommand{\ag}[1]{\textcolor{magenta}{[AG: #1]}}
\renewcommand{\hg}[1]{}
\renewcommand{\ag}[1]{}
\renewcommand{\sharon}[1]{}

\newcommand{\eqdef}{\triangleq}
\newcommand{\hole}{?}
\newcommand{\context}{Z}
\newcommand{\eq}{\approx}
\newcommand{\deq}{\not\approx}
\newcommand{\vr}{\vec{r}}
\newcommand{\vw}{\vec{w}}

\renewcommand{\gets}{\mathbin{:=}}

\newcommand{\pvV}{\pv{V}}
\newcommand{\pvv}{\pv{v}}
\newcommand{\pvx}{\pv{x}}
\newcommand{\pvy}{\pv{y}}

\newcommand{\pvc}{\pv{c}}

\newcommand{\exba}{\alpha_b}

\DeclareMathOperator{\base}{base}

\DeclareMathOperator{\depth}{depth}

\newcommand{\TS}{\mathcal{S}}
\newcommand{\Conf}{C}
\newcommand{\initconf}{c_0}
\newcommand{\Tr}{\mathcal{R}}
\newcommand{\tuple}[1]{\langle #1 \rangle}

\newcommand{\genAbs}{\sharp}
\newcommand{\ConfGenAbs}{\Conf}
\newcommand{\toGenAbs}{\to^{\genAbs}}
\newcommand{\TSGenAbs}{\genAbs(\TS)}
\newcommand{\TrGenAbs}{\Tr^{{\genAbs}}}
\newcommand{\initconfGenAbs}{\initconf^{{\genAbs}}}

\newcommand{\bbase}{base_\beta}
\newcommand{\rn}{\mathit{r}}
\newcommand{\cover}{\mathbb{C}}
\newcommand{\loc}{\mathcal{L}}
\newcommand{\limp}{\Rightarrow}
\newcommand{\inv}{\eta}
\newcommand{\Reach}{\mathit{Reach}} 
\newcommand{\qinit}{q_0}

\newcommand{\peuf}{\mathcal{P}_{\euf}}

\theoremclass{lemma}
\theorempreskip{1pt}
\theorempostskip{1pt}
\theoremsymbol{}
\newtheorem{slemma}[lemma]{Lemma}

\theoremclass{lemma}
\newtheorem{alemma}{Lemma}[section]

\title{Logical Characterization of Coherent Uninterpreted Programs}
\author{
\IEEEauthorblockN{Hari Govind V K}
\IEEEauthorblockA{\textit{University of Waterloo}}
\and
\IEEEauthorblockN{Sharon Shoham}
\IEEEauthorblockA{\textit{Tel-Aviv University}}
\and
\IEEEauthorblockN{Arie Gurfinkel}
\IEEEauthorblockA{\textit{University of Waterloo}}
}
\date{May 2021}

\begin{document}
\pagestyle{plain}
\Crefname{algorithm}{Alg.}{Alg.}
\crefname{algorithm}{Alg.}{Alg.}
\Crefname{section}{Sec.}{Sec.}
\crefname{section}{Sec.}{Sec.}
\Crefname{figure}{Fig.}{Fig.}
\crefname{figure}{Fig.}{Fig.}
\Crefname{equation}{Eq.}{Eq.}
\crefname{equation}{Eq.}{Eq.}
\Crefname{definition}{Def.}{Def.}
\crefname{definition}{Def.}{Def.}
\Crefname{lemma}{Lm.}{Lm.}
\crefname{lemma}{Lm.}{Lm.}
\Crefname{slemma}{Lm.}{Lm.}
\crefname{slemma}{Lm.}{Lm.}
\Crefname{alemma}{Lm.}{Lm.}
\crefname{alemma}{Lm.}{Lm.}
\Crefname{theorem}{Thm.}{Thm.}
\crefname{theorem}{Thm.}{Thm.}
\Crefname{corollary}{Cor.}{Cor.}
\crefname{corollary}{Cor.}{Cor.}

\maketitle
\begin{abstract}
An uninterpreted program (UP) is a program whose semantics is defined over the theory of uninterpreted functions. This is a common abstraction used in equivalence checking, compiler optimization, and program verification. While simple, the model is sufficiently powerful to encode counter automata, and, hence, undecidable. Recently, a class of UP programs, called coherent, has been proposed and shown to be decidable. We provide an alternative, logical characterization, of this result. Specifically, we show that every coherent program is bisimilar to a finite state system. Moreover, an inductive invariant of a coherent program is representable by a formula whose terms are of depth at most 1. We also show that the original proof, via automata, only applies to programs over unary uninterpreted functions. 
While this work is purely theoretical, it suggests a novel abstraction that is complete for coherent programs but can be soundly used on \emph{arbitrary} uninterpreted (and partially interpreted) programs.
\end{abstract}
\section{Introduction}
\label{sec:intro}

The theory of Equality with Uninterpreted Functions (EUF) is an important
fragment of First Order Logic, defined by a set of functions,
equality axioms, and congruence axioms. Its satisfiability problem is decidable.
It is a core theory of most SMT solvers, used as a glue (or abstraction) for
more complex theories. A closely related notion is that of Uninterpreted
Programs (UP), where all basic operations are defined by uninterpreted
functions. Feasibility of a UP computation is characterized by satisfiability of its path condition in EUF.
UPs provide a natural abstraction layer
for reasoning about software. They have been used (sometimes without explicitly
being named), in equivalence checking of pipelined
microprocesors~\cite{DBLP:conf/cav/BurchD94}, and equivalence checking of C
programs~\cite{DBLP:conf/vstte/StrichmanG05}. They also provide the foundations
of Global Value Numbering (GVN) optimization in many modern
compilers\cite{DBLP:conf/popl/Kildall73,DBLP:conf/sas/GulwaniN04,dblp:conf/vmcai/muller-olmrs05}.

Unlike EUF, reachability in UP is undecidable. That is, in the \emph{lingua franca} of SMT,
the satisfiability of Constrained Horn Clauses over
EUF is undecidable. Recently, Mathur et
al.~\cite{DBLP:journals/pacmpl/MathurMV19}, have proposed a variant of UPs,
called \emph{coherent uninterpreted program} (CUPs). The precise definition of
coherence is rather technical (see Def.~\ref{def:cup}), but intuitively the
program is restricted from depending on arbitrarily deep terms. The key result
of~\cite{DBLP:journals/pacmpl/MathurMV19} is to show that both reachability of
CUPs and deciding whether an UP is coherent are decidable. This makes CUP an
interesting infinite state abstraction with a \emph{decidable} reachability
problem.

Unfortunately, as shown by our counterexample in \Cref{fig:cup} (and described in Sec.~\ref{sec:char}), the
key construction in~\cite{DBLP:journals/pacmpl/MathurMV19} is incorrect. More
precisely, the proofs of~\cite{DBLP:journals/pacmpl/MathurMV19} hold only of
CUPs restricted to unary functions. In this paper, we address this bug. 
We provide an alternative (in our view simpler) proof of decidability and extend the results from reachability to arbitrary model checking. The case of
 non-unary CUPS is much more complex than unary. This is not surprising, since
similar complications arise in related results on Uniform Interpolation~\cite{DBLP:conf/cilc/GhilardiGK20} and
Cover~\cite{DBLP:conf/esop/GulwaniM08} for EUF.

Our key result is a logical characterization of CUP. We show that the set of
reachable states (i.e., the strongest inductive invariant) of a CUP is definable
by an EUF formula, over program variables, with terms of depth at most 1. That is, the
most complex term that can appear in the invariant is of the form $v \eq
f(\Vec{w})$, where $v$ and $\Vec{w}$ are program variables, and $f$ a function.

This characterization has several important consequences since the number of
such bounded depth formulas is finite. Decidability of reachability, for example, follows
trivially by enumerating all possible candidate inductive invariants. More importantly from a practical perspective, it leads to an efficient
analysis of \emph{arbitrary} UPs. Take a UP $P$, and check whether it has a safe
inductive invariant of bounded terms. Since the number of terms is finite, this
can be done by implicit predicate abstraction~\cite{DBLP:conf/tacas/CimattiGMT14}. If no invariant is found,
and the counterexample is not feasible, then $P$ is not a CUP. At this point, the
process either terminates, or another verification round is done with predicates
over deeper terms. Crucially, this does not require knowing whether $P$ is a CUP
apriori -- a problem that itself is shown in~\cite{DBLP:journals/pacmpl/MathurMV19} to
be at least PSPACE.

We extend the results further and show that CUPs are bisimilar to a finite
state system, showing, in particular, that arbitrary model checking for
CUP (not just reachability) is decidable.

Our proofs are structured around a series of abstractions, illustrated in a commuting diagram in \cref{fig:cd}. Our key abstraction is the base abstraction $\alpha_b$. It forgets terms deeper than depth 1, while maintaining all their consequences (by using additional fresh variables).  We show that $\alpha_b$ is sound and complete (i.e., preserves all properties) for CUPs (while, sound, but not complete for UP).  It is combined with a cover abstraction $\alpha_\cover$, that we borrow from~\cite{DBLP:conf/esop/GulwaniM08}. The cover abstraction ensures that reachable states are always expressible over program variables. It serves the purpose of existential quantifier elimination, that is not available for EUF. Finally, a renaming abstraction $\alpha_r$ is a technical tool to bound the occurrences of constants in abstract reachable states.

The rest of the paper is structured as follows. We review the necessary background on EUF in \cref{sec:background}. 
We introduce our formalization of UPs and CUPs in \cref{sec:up}. \Cref{sec:abs} presents bisimulation inducing abstractions for UP. \Cref{sec:extabase} presents our base abstraction and shows that it induces a bisimulation for CUPs. \Cref{sec:char} develops logical characterization for CUPs, presents our decidability results, and shows that a finite state abstraction of CUPs is computable. We conclude the paper in \cref{sec:conclusion} with summary of results and a discussion of open challenges and future work.

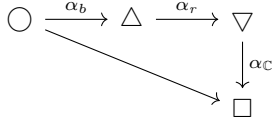
\begin{figure}
  \centering
\begin{tikzcd}
  \bigcirc \arrow{2-3} \arrow{r}{\exba} &
  \bigtriangleup \arrow{r}{\alpha_\rn}& \bigtriangledown \arrow{d}{\alpha_\cover}\\%
   & &\square
\end{tikzcd}
\caption{Sequence of abstractions used in our proofs. \label{fig:cd}
}
\end{figure}
\section{Background}
\label{sec:background}

We assume that the reader is familiar with the basics of First Order Logic
(FOL), and the theory of Equality and Uninterpreted Functions (EUF). We use
$\Sigma = (\consts, \funs, \{\eq, \deq\})$ to denote a FOL signature with
constants $\consts$, functions $\funs$, and predicates $\{\eq, \deq\}$,
representing equality and disequality, respectively. A term is a constant or
(well-formed) application of a function to terms. A literal is either $x \eq y$
or $x \deq y$, where $x$ and $y$ are terms. A formula is a Boolean combination
of literals. We assume that all formulas are quantifier free unless stated
otherwise. We further assume that all formulas are in Negation Normal Form
(NNF), so negation is defined as a shorthand: $\neg (x \eq y) \eqdef x \deq y$,
and $\neg (x \deq y) \eqdef x \eq y$. Throughout the paper, we use $\bowtie$ to
indicate a predicate in $\{\eq, \deq\}$. For example, $\{x \bowtie y\}$ means
$\{x \eq y, x \deq y\}$. We write $\bot$ for false, and $\top$ for true. We do
not differentiate between sets of literals $\Gamma$ and their conjunction
$(\bigwedge \Gamma)$. We write $\depth(t)$ for the maximal depth of function
applications in a term $t$. We write $\terms(\varphi)$, $\consts(\varphi)$, and
$\funs(\varphi)$ for the set of all terms, constants, and functions, in
$\varphi$, respectively, where $\varphi$ is either a formula or a collection of
formulas. Finally, we write $t[x]$ to mean that the term $t$ contains $x$ as a subterm. 

For a formula $\varphi$, we write $\Gamma \models \varphi$ if $\Gamma$
\emph{entails} $\varphi$, that is every model of $\Gamma$ is also a model of
$\varphi$. For any literal $\ell$, we write $\Gamma \vdash \ell$, pronounced
$\ell$ is \emph{derived} from $\Gamma$, if $\ell$ is derivable from
$\Gamma$ by the usual EUF proof system $\mathcal{P}_{\euf}$.\footnote{Shown in
  Appendix~\ref{sec:euf_extra}.} By refutational completeness of $\peuf$,
$\Gamma$ is unsatisfiable iff $\Gamma \vdash \bot$.



Given two \euf formulas $\varphi_1$ and $\varphi_2$ and a set of constants $V \subseteq \consts$,
we say that the formulas are $V$-equivalent, denoted $\varphi_1 \equiv_V
\varphi_2$, if, for all quantifier free \euf formulas $\psi$ such that
$\const(\psi) \subseteq V$, $(\varphi_1 \wedge \psi) \models \bot$ if and only
if $(\varphi_2 \wedge \psi) \models \bot$.

\begin{example}
  Let $\varphi_1 = \{x_1 \eq f(a_0, x_0), y_1 \eq f(b_0, y_0), x_0 \eq y_0\}$,
  $\varphi_2 = \{x_1 \eq f(a_0, w), y_1 \eq f(b_0, w)\}$, $\varphi_3 = \{x_1 \eq
  f(a_0, x_0), y_1 \eq f(b_0, y_0)\}$, and $V = \{x_1, y_1, a_0, b_0\}$. Then,
  $\varphi_1 \equiv_V \varphi_2$ but $\varphi_1 \not\equiv_V \varphi_3$.
\end{example}
While EUF does not admit quantifier elimination, it does admit elimination of
constants while preserving quantifier free consequences. Formally, a
\emph{cover}~\cite{DBLP:conf/esop/GulwaniM08,
  DBLP:conf/cade/CalvaneseGGMR19,DBLP:conf/cilc/GhilardiGK20}
of an EUF formula $\varphi$ w.r.t. a set of constants $V$ is an EUF formula
$\psi$ such that $\consts(\psi) \subseteq \consts(\varphi) \setminus V$ and
$\varphi \equiv_{\consts(\varphi) \setminus V} \psi$.
By~\cite{DBLP:conf/esop/GulwaniM08}, such $\psi$ exists and is unique up to
equivalence; we denote it by $\cover V \cdot \varphi$.



\section{Uninterpreted Programs}
\label{sec:up}

An \emph{uninterpreted program (UP)} is a program in the \emph{uninterpreted
  programming language (UPL)}. The \emph{syntax} of UPL is shown in
\Cref{fig:syntax}. Let $\pvV$ denote a fixed set of program variables.
We use lower case letters in a special font: $\pvx$, $\pvy$, etc. to denote
individual variables in $\pvV$. We write $\Vec{\pvy}$ for a list of program
variables. Function symbols are taken from a fixed set $\funs$.
As in~\cite{DBLP:journals/pacmpl/MathurMV19}, w.l.o.g., UPL does not
allow for Boolean combination of conditionals and relational symbols.


\begin{figure}[t]
\begin{align*}
    \langle stmt \rangle ::=\; &\mathbf{skip} \mid \langle var \rangle \gets \langle var \rangle \mid \langle var \rangle \gets f(\Vec{\langle var \rangle}) \mid \\
               &\mathbf{assume}\;(\langle cond \rangle) \mid \langle stmt \rangle \mathbin{;} \langle stmt \rangle \mid \\
               &\mathbf{if}\;(\langle cond \rangle ) \;\mathbf{then} \;\langle stmt \rangle  \;\mathbf{else}\; \langle stmt \rangle \mid \\
               &\textbf{while}\;(\langle cond \rangle) \;\langle stmt \rangle\\
    \langle cond\rangle ::=\; &\langle var \rangle = \langle  var \rangle \mid \langle var  \rangle \neq \langle  var \rangle\\
  \langle var \rangle ::=\; & \pvx \mid \pvy \mid \cdots
\end{align*}
\caption{Syntax of the programming language UPL.}
\label{fig:syntax}
\end{figure}

The small step symbolic operational semantics of UPL is defined with respect to
a FOL signature $\Sigma = (\consts, \funs, \{\eq, \deq\})$ by the rules shown in
Fig.~\ref{fig:ssesem}. A program \emph{configuration} is a triple $\langle s, q,
pc \rangle$, where $s$, called a statement, is a UP being executed, $q :
\pvV \to \consts$ is a \emph{state} mapping program variables to constants in
$\const$, and $pc$, called the \emph{path condition}, is a EUF formula
over $\Sigma$.
We use $\consts(q) \eqdef \{c \mid \exists \pvv \cdot q(\pvv)
= c\}$ to denote the set of all constants that represent current variable
assignments in $q$. With abuse of notation, we use $\const(q)$ and $q$
interchangebly. We write $\equiv_q$ to mean $\equiv_{\const(q)}$.

\renewcommand{\final}{\mathbf{skip}}
\begin{figure}
\centering
\begin{gather*}
\langle \mathbf{skip}\mathbin{;}s, q, pc \rangle \to \langle s, q, pc \rangle\\[1ex]
\prftree{\langle s_1, q, pc\rangle \to \langle s_1', q', pc'\rangle}%
{\langle s_1\mathbin{;}s_2, q, pc \rangle \to \langle s_1'\mathbin{;}s_2, q', pc' \rangle}\\[1ex]
\prftree{\langle c, q\rangle\Downarrow v\qquad}{(pc \land v) \not \models \bot}%
{\langle \mathbf{assume}(c), q, pc \rangle \to \langle \final, q, pc\wedge v\rangle}\\[1ex]
\prftree{\langle e, q\rangle \Downarrow v\qquad}{\text{$x' \in \consts(\Sigma)$ is fresh in $pc$}}{\langle \pvx \gets e, q, pc \rangle \to \langle\final, q[\pvx \mapsto x'], pc \land x' = v\rangle}\\[1ex]
\begin{aligned}
\langle \mathbf{if}\;(c)\;\mathbf{then}\;s_1\;\mathbf{else}\;s_2, q, pc \rangle &\to \langle \mathbf{assume}(c)\mathbin{;}s_1, q, pc\rangle\\[1ex]
\langle \mathbf{if}\;(c)\;\mathbf{then}\;s_1\;\mathbf{else}\;s_2, q, pc \rangle &\to \langle \mathbf{assume}(\neg c)\mathbin{;}s_2, q, pc\rangle
\end{aligned}\\[1ex]
\begin{multlined}
    \langle \mathbf{while}\;(c)\;s, q, pc \rangle \to{} \\ \qquad\qquad \langle \mathbf{if}\;(c)\;\mathbf{then}\;(s\mathbin{;} \mathbf{while}\;(c)\;s)\;\mathbf{else}\;\mathbf{skip}, q, pc \rangle
\end{multlined}
\end{gather*}
\caption{Small step symbolic operational semantics of UPL, where $\neg c$ denotes $\pvx \neq \pvy$ when $c$ is $\pvx = \pvy$, and
$\pvx = \pvy$ when $c$ is $\pvx \neq \pvy$.}
\label{fig:ssesem}
\end{figure}
For a state $q$, we write $q[\pvx \mapsto x']$ for a state $q'$ that is
identical to $q$, except that it maps $\pvx$ to $x'$.
We write $\langle e, q
\rangle \Downarrow v$ to denote that $v$ is the value of the expression $e$ in
state $q$, i.e., the result of substituting each program variable $\pv{x}$ in
$e$ with $q(\pv{x})$, and replacing functions and predicates with their FOL
counterparts.
The value of $e$ is an FOL term or an FOL formula over $\Sigma$.
For example, $\langle \pvx = \pvy, [\pvx \mapsto x, \pvy \mapsto y]
\rangle \Downarrow x \eq y$. 

Given two configurations $c$ and $c'$, we write $c
\to c'$ if $c$ reduces to $c'$ using one of the rules in \Cref{fig:ssesem}.
Note that there is no rule for $\textbf{skip}$ -- the program terminates once it gets into a configuration $\langle \textbf{skip}, q, pc\rangle$.

Let $\consts_0 = \{v_0 \mid \pvv \in \pvV\} \subseteq \consts$ be a set of initial constants. In
the initial state $\qinit$ of a program, every variable is mapped to
the corresponding initial constant, i.e., $\qinit(\pvv) = v_0$.

The operational semantics induces, for an UP $P$, a transition system $\TS_P =
\tuple{\Conf,\initconf,\Tr}$, where $\Conf$ is the set of configurations,
$\initconf \eqdef \langle P, \qinit, \top \rangle$ is the initial
configuration, and $\Tr \eqdef \{(c, c') \mid c \to c'\}$.
A configuration $c$ of $P$ is \emph{reachable} if $c$ is reachable from
$\initconf$ in $\TS_P$. 
We denote the set of all reachable configurations in $\TS_P$ using $\Reach(\TS_P)$. The set of all statements in the semantics of $P$, including the intermediate statements, are called
\emph{locations} of $P$, and are denoted by $\loc(P)$.
We often use $P$ and $\TS_P$ interchangeably. 

Our semantics of UPL differs in some respects from the one
in~\cite{DBLP:journals/pacmpl/MathurMV19}. First, we follow a more traditional
small-step operational semantics presentation, by providing semantics rules and
the corresponding transition system. However, this does not change the semantics
conceptually. More importantly, we ensure that the path condition remains
satisfiable in all reachable configurations (by only allowing an assume statement to execute when
it results in a satisfiable path condition). We believe this is a more natural
choice that is also consistent with what is typically used in other symbolic
semantics.
UP reachability under our semantics coincides with the definition of~\cite{DBLP:journals/pacmpl/MathurMV19}.
\begin{definition}[UP Reachability]
  \label{def:reach}
  \hg{modified}
   Given an UP $P$, determine whether there exists a state $q$ and a path condition $pc$ s.t., the configuration
  $\langle\textbf{skip}, q, pc\rangle$ is reachable in $P$.
\end{definition}

A certificate for unreachability of location $s$, is an inductive assertion map $\inv$ (or an inductive invariant) s.t. $\inv(s) = \bot$.

\begin{definition}[Inductive Assertion Map]
  \label{def:inductive}
  Let $\Sigma_0 \eqdef (\consts_0, \funs, \{\eq, \deq\})$, be restriction of $\Sigma$ to $\consts_0$.
  An \emph{inductive assertion map} of an UP $P$, is
  a map $\inv : \loc(P) \to EUF(\Sigma_0)$ s.t. (a) $\inv(P) = \top$,
  and (b) if $\langle s, \qinit, \inv(s) \rangle \to \langle s', q', pc' \rangle$,
  then $pc' \models (\inv(s')[v_0 \mapsto q'(\pvv) \mid \pvv \in \pvV])$.
\end{definition}


In~\cite{DBLP:journals/pacmpl/MathurMV19}, a special sub-class of UPs has been introduced with a decidable reachability problem.

\begin{definition}[Coherent Uninterpreted Program~\cite{DBLP:journals/pacmpl/MathurMV19}]
  \label{def:cup}%
An UP $P$ is \emph{coherent} (CUP) if all of the reachable
configurations of $P$ satisfy the following two properties:
\begin{description}[topsep=0pt,noitemsep]
\item[Memoizing] for any configuration $\langle \pvx \gets f(\Vec{\pvy}), q,
  pc \rangle$, if there is a term $t \in \terms(pc)$ s.t. $pc \models t
  \eq f(q(\Vec{\pvy}))$, then there is $\pvv \in \pvV$ s.t. $pc \models
  q(\pvv) \eq t$. 
\item[Early assume] for any configuration\\\mbox{$\langle \mathbf{assume}(\pvx =
    \pvy), q, pc \rangle$}, if there is a term $t \in \terms(pc)$ s.t. $pc
  \models t \eq s$ where $s$ is a superterm of either $q(\pvx)$ or $q(\pvy)$,
  then, there is $\pvv \in \pvV$ s.t. $pc \models q(\pvv) \eq t$. 
\end{description}
\end{definition}
 Intuitively, memoization ensures that if a term is recomputed, then it is already stored in a program variable; early assumes ensures that whenever an equality between variables is assumed, any of their superterms that was ever computed is still stored in a program variable. Note that unlike the original definition of CUP in~\cite{DBLP:journals/pacmpl/MathurMV19}, we do not require the notion of an
\emph{execution}. The path condition accumulates the history of the execution in
a configuration, which is sufficient. 

\begin{figure}
\centering
\begin{subfigure}[t]{0.2\textwidth}
  \begin{lstlisting}[numbers=left]
  x := t;
  y := t;
  while (c != d) {
    x := n(x);
    y := n(y);
    c := n(c);
  };
  x := f(a, x);
  y := f(b, y);
  assume(a == b);
  assume(x != y);
  \end{lstlisting}
\end{subfigure}
\begin{subfigure}[t]{0.2\textwidth}
\[\footnotesize
    \begin{array}{l}
      x_0 \eq t_0\\
      x_0 \eq t_0 \land y_0 \eq t_0\\
      x_0 \eq y_0\\
      x_0 \eq n(y_0) \land c_0 \deq d_0\\
      x_0 \eq y_0 \land c_0 \deq d_0\\
      x_0 \eq y_0\\
      \\
      x_0 \eq f(a_0, y_0) \land c_0 \eq d_0\\
      (a_0 \eq b_0 \limp x_0 \eq y_0) \land c_0 \eq d_0\\
      a_0 \eq b_0 \land x_0 \eq y_0 \land c_0 \eq d_0 \\
      \bot
    \end{array}
\]
\end{subfigure}
\caption{An example CUP program and its inductive assertions.}
\label{fig:cup}
\end{figure}
\begin{example}\NoEndMark
An example of a CUP is shown in Fig.~\ref{fig:cup}. Some reachable states in the first iteration of the loop are shown below, where line numbers are used as locations, and $pc_i$ stands for the path condition at line $i$:
\[
\begin{lgathered}
\langle 2, q_0[\pvx \mapsto x_1, \pvy \mapsto y_1], x_1 \eq t_0 \land y_1 \eq t_0\rangle\\
\begin{multlined}[t]
      \langle 6, q_0[\pvx \mapsto x_2, \pvy \mapsto y_2, \pvc \mapsto c_1], pc_2
      \land{}\\ \quad c_0 \deq d_0 \land x_2 \eq n(x_1) \land y_2 \eq n(y_1) \land c_1 \eq n(c_0)\rangle
\end{multlined}\\
\begin{multlined}[t]
      \langle 9, q_0[\pvx \mapsto x_3, \pvy \mapsto y_3, \pvc \mapsto c_1]\rangle,
       pc_6 \land{}\\ \qquad\qquad c_1 \eq d_0 \land x_3 \eq f(a_0, x_2) \land y_3 \eq
                                      f(b_0, y_2) \rangle
\end{multlined}
\end{lgathered}
\]
The program is coherent because (a)~no term is recomputed; (b)~for the assume at line~10, the only superterms of $a_0$ and $b_0$ are $f(a_0, x_n)$ and $f(b_0, y_n)$, and they are stored in $\pvx$ and $\pvy$, respectively; and \hg{added}(c)~for the assume $(c_n = d_0)$ introduced by the exit condition of the while loop, no superterms of $c_n$, $d_0$ are ever computed.
The program does not reduce to 
$\final$ (i.e., it does not reach a final configuration). Its inductive assertion map is shown in
Fig.~\ref{fig:cup} (right).
\hfill\ExampleSymbol

\end{example}

Note that UP are closely related, but are not equivalent, to the
Herbrand programs of~\cite{dblp:conf/vmcai/muller-olmrs05}. While Herbrand
programs use the syntax of UPL, they are interpreted over a fixed universe of
Herbrand terms. In particular, in Herbrand programs $f(x) \eq g(x)$ is always
false (since $f(x)$ and $g(x)$ have different top-level functions),
while in UP, it is satisfiable. 
%

\section{Abstraction and Bisimulation for UP}
\label{sec:abs}
In this section, we review abstractions for transition systems. We then define two abstraction for UP: cover and renaming, and show that they induce bisimulation. That is, for UP, these abstractions preserve all properties. Finally, we show a simple logical characterization result for UP to set the stage for our main results in the following sections.

\begin{definition}
\label{def:abstract-TS-general}
Given a transition system $\TS = (\Conf, \initconf,\Tr)$
and a (possibly partial) abstraction function $\genAbs: \Conf \to \ConfGenAbs$, 
the induced \emph{abstract transition system}
is $\TSGenAbs = (\ConfGenAbs, \initconfGenAbs, \TrGenAbs)$, where
\begin{align*}
\initconfGenAbs &\eqdef \genAbs(\initconf) \\
\TrGenAbs &\eqdef \{(c_\sharp,c_\sharp') \mid \exists c, c'.~c \to c'~\land~c_\sharp = \genAbs(c)~\land~c_\sharp' = \genAbs(c') \}
\end{align*}
We write $c \toGenAbs c'$ when $(c,c') \in \TrGenAbs$. Note that $\genAbs$ must be defined for $\initconf$.
\end{definition}

Throughout the paper, we construct several abstract transition systems. All transition systems considered are \emph{attentive}. Intuitively, this means that their transitions do not distinguish between configurations that have $q$-equivalent path conditions. We say that two configurations $c_1 = \langle s, q, pc_1\rangle$ and $c_2 = \langle s, q, pc_2\rangle$ are equivalent, denoted $c_1 \equiv c_2$ if $pc_1 \equiv_q pc_2$. 
\begin{definition}[Attentive TS]
\label{def:attentive}
A transition system $\TS = (\Conf, \initconf, \Tr)$ is \emph{attentive} if for any two configurations $c_1, c_2 \in \Conf$ s.t. $c_1 \equiv c_2$, if there exists $c_1'\in \Conf$ s.t. $(c_1, c_1') \in \Tr$, then there exists $c_2'\in \Conf$, s.t. $(c_2, c_2')\in \Tr$ and $c_1' \equiv c_2'$ and vice versa. 
\end{definition}

Weak, respectively strong, preservation of properties between the abstract and the concrete transition systems
are ensured by the notions of \emph{simulation}, respectively \emph{bisimulation}.

\begin{definition}[\cite{DBLP:books/daglib/0067019}]
Let $\TS = (\Conf, \initconf, \Tr)$ and $\TSGenAbs = (\ConfGenAbs, \initconfGenAbs, \TrGenAbs)$
be transition systems. A relation  $\rho \subseteq \Conf \times \ConfGenAbs$ is a \emph{simulation} from $\TS$ to $\TSGenAbs$,
if for every $(c,c_\sharp) \in \rho$:
\begin{itemize}
\item if $c \to c'$ then there exists $c_\sharp'$ such that $c_\sharp \toGenAbs c_\sharp'$ and $(c',c_\sharp') \in \rho$.
\end{itemize}
$\rho \subseteq \Conf \times \ConfGenAbs$ is a \emph{bisimulation} from $\TS$ to $\TSGenAbs$ if
$\rho$ is a simulation from $\TS$ to $\TSGenAbs$ and
$\rho^{-1}\eqdef \{(c_\sharp,c) \mid (c,c_\sharp) \in \rho \}$ is a simulation from $\TSGenAbs$ to $\TS$.
We say that $\TSGenAbs$ \emph{simulates}, respectively \emph{is bisimilar to}, $\TS$ if there exists
a simulation, respectively, a bisimulation, $\rho$ from $\TS$ to $\TSGenAbs$ such that $(\initconf, \initconfGenAbs) \in \rho$.
\end{definition}

We say that a bisimulation $\rho \subseteq \Conf \times \ConfGenAbs$ is \emph{finite} if its range, $\{ \rho(c) \mid c \in \Conf \}$, is finite. A finite bisimulation relates a (possibly infinite) transition system with a finite one.

Next, we define two abstractions for UP programs and show that they result in bisimilar abstract transition systems.
The first abstraction eliminates all constants that are not assigned to program variables from the path condition, using the cover operation.
The second abstraction renames the constants assigned to program variables back to the initial constants $\consts_0$.
Both abstractions together ensure that all reachable configurations in the abstract transition system are defined over $\Sigma_0$ (i.e., the only constants that appear in states, as well as in path conditions, are from $\consts_0$). There may still be infinitely many such configurations since the depth of terms may be unbounded.
We show that whenever the obtained abstract transition system has finitely many reachable configurations, the concrete one has an inductive assertion map that characterizes the set of reachable configurations.
\begin{definition}[Cover abstraction]
\label{def:cover-abs}
The cover abstraction function $\alpha_\cover: \Conf \to \Conf$ is defined by
\[
\alpha_\cover(\langle s, q, pc\rangle) \eqdef \langle s,q,\cover (\consts \setminus \consts(q)) \cdot pc\rangle
\]
\end{definition}
Since $pc \equiv_{q} \cover (\consts \setminus \consts(q)) \cdot pc$, the cover
abstraction also results in a bisimilar abstract transition system.
\begin{theorem}
\label{thm:cover-bisimilar}
For any attentive transition system $\TS = (\Conf, \initconf, \Tr)$, the relation $\rho =
\{(c, \alpha_\cover(c)) \mid c \in \Reach(\TS)\}$ is a bisimulation from $\TS$ to
$\alpha_\cover(\TS)$.
\end{theorem}


To introduce the renaming abstraction, we need some notation.
Given a quantifier free formula $\varphi$, constants $a, b \in \const(\varphi)$ such that $a\neq b$, let $\varphi[a \rename b]$ denote $\varphi[b\mapsto x][a \mapsto b]$, where $x$ is a constant not in $\const(\varphi)$. For example, if $\varphi = (a \eq c \land b \eq d)$, $\varphi[a\rename b] = (b \eq c \wedge x \eq d)$.


Given a path condition $pc$ and a state $q$, let $\rn_{0}(pc, q)$ denote the
formula obtained by renaming all constants in $\const(q)$ using their initial
values. $\rn_{0}(pc, q) = pc[q(\pv{v}) \rename v_0]$ for all $\pvv \in \pvV$ such
that $q(\pv{v}) \neq v_0$.
\begin{definition}[Renaming abstraction]
\label{def:rename-abstraction}
The renaming abstraction function $\alpha_r: \Conf \to \Conf$ is defined by
\[
\alpha_r(\langle s, q, pc\rangle) \eqdef \langle s,q_{0},\rn_{0}(pc, q)\rangle
\]
\end{definition}
\begin{theorem}
\label{thm:renaming-bisimilar}
For any attentive transition system $\TS = (\Conf, \initconf, \Tr)$, the relation $\rho =
\{(c, \alpha_r(c)) \mid c \in \Reach(\TS)\}$ is a bisimulation from $\TS$ to
$\alpha_{\rn}(\TS)$.
\end{theorem}
Finally, we denote by $\alpha_{\cover,\rn}$ the composition of the renaming and
cover abstractions: $\alpha_{\cover,\rn} \eqdef \alpha_\cover \circ \alpha_\rn$
(i.e., $\alpha_{\cover,\rn}(c) = \alpha_\rn(\alpha_\cover(c))$). Since the
composition of bisimulation relations is also a bisimulation,
$\alpha_{\cover, \rn}(\TS)$ is bisimilar to $\TS$.

\begin{theorem}[Logical Characterization of UP]\label{th:lcup}
  If $\alpha_{\cover,\rn}$ induces a finite bisimulation on an UP $P$, then, there
  exists an inductive assertion map $\inv$ for $P$ that characterizes the reachable configurations of $P$.
\end{theorem}
\begin{proof}
  Define $\inv(s) \eqdef \bigvee \{ pc \mid \langle s, q, pc \rangle \in
  \Reach(\alpha_{\cover,\rn}(P))\}$. Then, $\inv(s)$ is such an inductive
  assertion map.
\end{proof}
%

Intuitively, \cref{th:lcup} says that inductive invariant of UP, whenever it exists, can be described using EUF formulas over program variables. That is, any extra variables that are added to the path condition during program execution can be abstracted away (specifically, using the cover abstraction). There are, of course, infinitely many such invariants since the depth of terms is not bounded (only constants occurring in them). 
In the sequel, we systematically construct a similar result for CUP.
\section{Bismulation of CUP}
\label{sec:extabase}

The first step in extending~\cref{th:lcup} to CUP is to design an abstraction function 
that bounds the depth of terms that appear in any reachable (abstract) state. It is easy to 
design such a function while maintaining soundness -- simply forget literals that have terms that are too 
deep. However, we want to maintain precision as well. That is, we want the abstract transition
system to be bisimilar to the concrete one. 
Just like cover abstraction, the base abstraction
function also eliminates all constants that are not assigned to program
variables. Unlike cover abstraction, the base abstraction does not maintain
$\consts(q)$-equivalence of the path conditions, but, rather, forgets most literals that cannot be expressed over program variables. 

In this section, we focus on the definition of the base abstraction and prove that it induces bisimulation for CUP. This result is used in \cref{sec:char}, to logically characterize CUPs.

Intuitively, the base abstraction ``truncates'' the congruence graph induced by
a path condition in nodes that have no representative in the set of constants
assigned to the program variables ($V$ in the following definition), and assigns
to the truncated nodes fresh constants (from $W$ in the following definition).

Congruence closure procedures for EUF use a \emph{congruence graph} to concisely
represent the deductive closure of a set of EUF
literals~\cite{DBLP:journals/jacm/NelsonO80,DBLP:conf/rta/NieuwenhuisO05}. Here,
we use a logical characterization of a congruence graph, called a
\emph{$V$-basis}. Let $\Gamma$ be a set of EUF literals. A triple $\langle W,
\beta, \delta \rangle$ is a $V$-basis of $\Gamma$ relative to a set of constants
$V$, written $\langle W, \beta, \delta \rangle \in \base(\Gamma, V)$, iff (a)
$W$ is a set of fresh constants not in $\consts(\Gamma)$, and $\beta$ and
$\delta$ are conjunctions of EUF literals; (b) ($\exists W\cdot \beta \land \delta) \equiv
\Gamma$; (c) $\beta \eqdef \beta_{\eq} \cup \beta_{\deq} \cup \beta_{\funs}$ and
$\delta \eqdef \delta_{\eq} \cup \delta_{\deq} \cup \delta_{\funs}$, where
\begin{align*}
  \beta_{\eq} &\subseteq \{u \eq v \mid u, v \in V\} \qquad \beta_{\deq}\subseteq \{u \deq v \mid u, v \in V\} \\
  \beta_{\funs} &\subseteq \{v \eq f(\Vec{w}) \mid v \in V, \Vec{w} \subseteq V \cup W, \Vec{w} \cap V \neq \emptyset\}\\
  \delta_{\eq} &\subseteq \{ w \eq u \mid w \in V \cup W, u \not \in V \cup W\}
  \\
  \delta_{\deq} &\subseteq \{u \deq w \mid u \in W, w \in W \cup V\}\\
  \delta_{\funs} &\subseteq \{v \eq f(\Vec{w}) \mid v,\Vec{w}\subseteq V\cup W, v \in V \limp \Vec{w} \subseteq W\}
\end{align*}
(d) $\beta \land \delta \nvdash v \eq w$ for any $v \in V$, $w \in W$; and 
(e) $\beta \land \delta \nvdash w_1 \eq w_2$ for any $w_1, w_2 \in W$ s.t. $w_1 \neq
w_2$. 

Note that we represent both equalities and disequalities in the $V$-basis as
common in implementations (but not in the theoretical presentations) of the
congruence closure algorithm. Intuitively, $V$ are constants in
$\consts(\Gamma)$ that represent equivalence classes in $\Gamma$, and $W$ are
constants added to represent equivalence classes that do not have a
representative in $V$. A $V$-basis, of any satisfiable set $\Gamma$, is unique up to renaming of constants in $W$
and ordering of equalities between constants in $V$.
\begin{example}
\label{ex:basis}
  Let $\Gamma = \{x \eq f(a, v_1), y \eq f(b, v_2), v_1 \eq v_2\}$ and $V = \{a,
  b, x, y\}$. A $V$-basis of $\Gamma$ is $\langle W, \beta, \delta\rangle$,
  where $W = \{w\}$, $\beta = \{ x \eq f(a, w), y \eq f(b, w)\}$, $\delta = \{w \eq v_1, w \eq v_2\}$. Renaming $w$ to $w'$ is a different $V$-basis:
  $\langle W', \beta', \delta'\rangle \in \base(\Gamma, V)$ where $W' = \{w'\}$,
  $\beta' = \beta[w \mapsto w']$ and $\delta' = \delta[w \mapsto w']$. 
  
  As another example, consider $\Gamma = \{x \eq f(a, p), x \eq f(a, n(p)), y = f(b, p), y = f(c, n(p))\}$ and $V =\{x, y, a, b, c\}$. A $V$-basis of $\Gamma$ is $\langle W, \beta, \delta\rangle$,
  where $W = \{w_0, w_1\}$,  $\delta_2 = \{w_0 \eq p, w_1 \eq n(w_0)\}$, and
  \[
  \beta_2 = \left\{\begin{aligned} 
               x &\eq f(a, w_0) & 
               x &\eq f(a, w_1) \\
               y &\eq f(b, w_0) &
               y &\eq f(c, w_1)
               \end{aligned}\right\}
  \]
\end{example}

 While a basis maintains all consequences of $\Gamma$ (since $(\exists W\cdot\beta \wedge \delta) \equiv \Gamma$), the $V$-base abstraction of $\Gamma$, defined next, is weaker. It preserves consequences of $\beta$ only: 
 
\begin{definition}[$V$-base abstraction]
\label{def:vabst}
The $V$-base abstraction $\alpha_V$ for a set of constants $V$, is a function between sets of literals s.t.
for any sets of literals $\Gamma$ and $\Gamma'$:
\begin{enumerate}[(1)]
    \item $\alpha_V(\Gamma) \eqdef \beta$,  where  $\langle W, \beta, \delta\rangle \in \base(\varphi, V)$,
    \item if there exists a $\beta$ s.t. $\langle W_1, \beta, \delta_1\rangle \in\base(\Gamma, V)$ and $\langle W_2, \beta, \delta_2\rangle \in \base(\Gamma', V)$, then
      $\alpha_V(\Gamma) = \alpha_V(\Gamma')$.
\end{enumerate}
\end{definition}
The second requirement of Def.~\ref{def:vabst} ensures that two formulas that have the same
$V$-consequences, have the same $V$-abstraction. For example,
for a set of constants $V = \{u, v\}$, the formulas $\varphi_1 = \{v \eq f(u,
x)\}$ and $\varphi_2 = \{v \eq f(u, y)\}$, have the same $V$-base abstraction:
$v \eq f(u, w)$. Note that at this point, we only require that $\alpha_V$ is well defined (for example, it does not have to be computable.)


We now extend $V$-base abstraction to program configuration, calling it simply \emph{base abstraction}, since the set of preserved constants is determined by the configuration:
\begin{definition}[Base abstraction]
  \label{def:alpha-abstraction}
  The base abstraction $\exba: \Conf \to \Conf$ is defined for
  configurations  $\langle s, q, pc \rangle \in \Conf$, where $pc$ is a
  \emph{conjunction} of literals: $\exba(\langle
  s, q, pc\rangle) \eqdef \langle s,q,\alpha_{\const(q)}(pc)\rangle$.
\end{definition}

Namely, the base abstraction $\alpha_{\const(q)}$ applied to the
path condition is determined by the state $q$ in the configuration. We often write $\alpha_q(\varphi)$ as a shorthand for
$\alpha_{\consts(q)}(\varphi)$.

We are now in position to state the main result of this section.
Given a CUP $P$, the abstract transition system
$\exba(\TS_P) = (\Conf, \initconf^\exba, \Tr^\exba)$ is bisimilar to the
concrete transition system $\TS_P = (\Conf, \initconf, \Tr)$. Note that at this point, we do not claim that
$\exba(\TS_P)$ is finite, or that it is computable. We focus only on the fact that the literals that are forgotten by the base abstraction do not matter for any future transitions. The key technical step is summarized in the following theorem:
\begin{theorem}
  \label{thm:main}
  Let $\langle s, q, pc\rangle$ be a reachable configuration of a CUP $P$. Then,
  \begin{enumerate}[(1)]
  \item $\langle s, q, pc\rangle \to \langle s', q', pc \land pc'\rangle$ iff\\
    $\langle s, q, \alpha_q(pc)\rangle \to \langle s', q', \alpha_{q}(pc) \land
    pc'\rangle$, and
  \item $\alpha_{q'}(pc \land pc') = \alpha_{q'}(\alpha_q(pc) \land pc')$.
  \end{enumerate}
\end{theorem}

The proof of Thm.~\ref{thm:main} is not complicated, but it is tedious and technical. It depends on many basic properties of EUF. We summarize the key results that we require in the following lemmas. The proofs of the lemmas are provided in App.~\ref{sec:proofs}.

We begin by defining a \emph{purifier} -- a set of constants  sufficient to represent a set of EUF literals with terms of depth one.
\begin{definition}[Purifier]
  \label{def:purifier}
  We say that a set of constants $V$ is a \emph{purifier} of a constant $a$ in
  a set of literals $\Gamma$, if $a \in V$ and for every term $t \in \terms(\Gamma)$ s.t. $\Gamma \vdash t \eq s[a]$,  $\exists v \in V$ s.t. $\Gamma \vdash v \eq t$.
\end{definition}

For example, if $\Gamma = \{ c \eq f(a), d \eq f(b), d \deq e\}$. Then, $V = \{a, b, c\}$ is  a purifier for $a$, but  not a purifier for $b$, even though $b \in V$.

In all the following lemmas,  $\Gamma$, $\varphi_1$, $\varphi_2$ are sets
of literals; $V$ a set constants; $a, b \in \const(\Gamma)$; $u, v, x, y \in V$; $V$ is a purifier for $\{x, y\}$ in $\Gamma$, $\varphi_1$, and in $\varphi_2$; $\beta = \alpha_V(\Gamma)$; and $\alpha_V(\varphi_1) = \alpha_V(\varphi_2)$.

Lemma~\ref{lm:sprtrm} says that anything newly derivable from $\Gamma$ and a new equality $a \eq b$ is derivable using superterms of $a$ and $b$:
\begin{slemma}\label{lm:sprtrm}
  Let $t_1$ and $t_2$ be two terms in $\terms(\Sigma)$ s.t. $\Gamma
  \not\vdash (t_1 \eq t_2)$. Then, $(\Gamma \land a \eq b) \vdash (t_1 \eq
  t_2)$, for some constants $a$ and $b$ in $\consts(\Gamma)$, iff there are two
  superterms, $s_1[a]$ and $s_2[b]$, of $a$ and $b$, respectively, s.t. (i)
  $\Gamma \vdash (t_1 \eq s_1[a])$, (ii) $\Gamma \vdash (t_2 \eq s_2[b])$, and
  (iii) $(\Gamma \land a \eq b ) \vdash (s_1[a] \eq s_2[b])$.
\end{slemma}
Lemma~\ref{lm:Visenough} and Lemma~\ref{lm:Visenoughdeq} say that all consequences of $\Gamma$ that are relevant to $V$ are present in $\beta = \alpha_V(\Gamma)$ as well.
\begin{slemma}\label{lm:Visenough}
  $(\Gamma \land x \eq y \vdash u \eq v) \iff (\beta \land x \eq y \vdash u \eq
  v) $.
\end{slemma}
\begin{slemma}\label{lm:Visenoughdeq}
$(\Gamma \land x \eq y \vdash u \deq v) \iff ( \beta \land x \eq y \vdash u \deq v)$.
\end{slemma}
Lemma~\ref{lm:pur} says that $\beta = \alpha_V(\Gamma)$ can be described using terms of depth one using constants in $V$.
\begin{slemma}\label{lm:pur}
  $V$ is a purifier for $x \in V$ in $\beta$.
\end{slemma}
Lemma~\ref{lm:idem} says that $\alpha_V$ is idempotent.
\begin{slemma}\label{lm:idem}
  $\alpha_V(\Gamma) = \alpha_V(\alpha_V(\Gamma))$.
\end{slemma}
Lemma~\ref{lm:eqpres1v} and Lemma~\ref{lm:subset} say that $\alpha_V$ preserves addition of new literals and dropping of constants.
\begin{slemma}\label{lm:eqpres1v}
   $\alpha_V(\varphi_1 \land x \eq y) = \alpha_V(\varphi_2 \land x \eq y)$.
\end{slemma}
\begin{slemma}\label{lm:subset}
   If $U \subseteq V$, then \[(\alpha_V(\varphi_1)= \alpha_V(\varphi_2)) \Rightarrow (\alpha_U(\varphi_1) = \alpha_U(\varphi_2))\]
\end{slemma}
Lemma~\ref{lm:deqpres1v} extends the preservation results to disequalities. $V$ is a set of constants, $x, y \in V$. $V$ is not required to be a purifier (as it was in the previous
 lemmas). 
\begin{slemma}\label{lm:deqpres1v}
 $\alpha_V(\varphi_1 \land x \deq y) = \alpha_V(\varphi_2 \land x \deq y)$.
\end{slemma}
Lemma~\ref{lm:propequiv} extends the preservation results for equalities involving a fresh constant
$x'$ s.t. $x' \not \in \const(\varphi_1) \cup
 \const(\varphi_2)$. $\Vec{y}\subseteq V$, $V' = V \cup\{x'\}$, and $f(\Vec{y})$ be a term s.t
  there does not exists a term $t \in \terms(\varphi_1) \cup \terms(\varphi_2)$
  s.t. $\varphi_1 \vdash t \eq f(\Vec{y})$ or $\varphi_2 \vdash t \eq
  f(\Vec{y})$.
\begin{slemma}\label{lm:propequiv}
  \begin{align*}
  \tag{1}\alpha_{V'}(\varphi_1 \land x' \eq y) &= \alpha_{V'}(\varphi_2 \land x' \eq y)\\
  \tag{2}\alpha_{V'}(\varphi_1 \land x' \eq f(\Vec{y})) &= \alpha_{V'}(\varphi_2 \land x' \eq f(\Vec{y}))
  \end{align*}
\end{slemma}

We are now ready to present the proof of Thm.~\ref{thm:main}:
\begin{proof}[Theorem~\ref{thm:main}]
  In the proof, we use $x = q(\pv{x})$, and $y =
  q(\pv{y})$.
  For part (1), we only show the proof for $s = \textbf{assume}(\pv{x} \bowtie
  \pv{y})$ since the other cases are trivial.

  The only-if direction follows since $\alpha_q(pc)$ is weaker than $pc$.
  For the if direction, $pc \not\vdash \bot$ since it is part of a reachable configuration. Then, there are two cases:
\begin{itemize}
  \item case $s = \textbf{assume}(\pv{x}=\pv{y})$.
  Assume $(pc \land x \eq y) \models \bot$. Then, $(pc \land x \eq y) \vdash t_1 \eq t_2$ and $pc \vdash t_1 \deq t_2$ for some $t_1, t_2 \in \terms(pc)$. By Lemma~\ref{lm:sprtrm}, in
    any new equality $(t_1 \eq t_2)$ that is implied by $pc \land (x \eq y)$
    (but not by $pc$), $t_1$ and $t_2$ are equivalent (in $pc$) to superterms of
    $x$ or $y$. By the early assume property of CUP, $\consts(q)$ purifies $\{x,
    y\}$ in $pc$. Therefore, every superterm of $x$ or $y$ is equivalent (in
    $pc$) to some constant in $\consts(q)$. Thus,  $(pc \land x \eq y) \vdash u \eq v$ and $(pc \land x \eq y)  \vdash u \deq v$ for some $u, v \in \consts(q)$.
    By Lemma~\ref{lm:Visenough}, $(\alpha_q(pc) \land x \eq y) \vdash u \eq v$.
    By Lemma~\ref{lm:Visenoughdeq}, $(\alpha_q(pc) \land x \eq y) \vdash u \deq v$. Thus, $(\alpha_q(pc)\land x \eq y) \models \bot$.
  \item case $s = \textbf{assume}(\pv{x}\neq\pv{y})$. $(pc \land x \deq y)
    \models \bot$ if and only if $pc \vdash x \eq y$. Since $x, y \in \consts(q)$,
    $\alpha_q(pc) \vdash x \eq y$.
\end{itemize}

For part (2), we only show the cases for assume and assignment statements, the
other cases are trivial.
\begin{itemize}
\item case $s = \textbf{assume}(\pv{x} = \pv{y})$, Since $q' = q$, we need to
  show that $\alpha_{q}(pc \land x \eq y) = \alpha_{q}(\alpha_q(pc)
  \land x \eq y)$. From the early assumes property, $\const(q)$
  purifies $\{x, y\}$ in $pc$.
  By Lemma~\ref{lm:pur}, $\consts(q)$ purifies $\{x, y\}$ in $\alpha_q(pc)$ as well.
  By Lemma~\ref{lm:idem}, $\alpha_q(pc) = \alpha_q(\alpha_q(pc))$.
  By Lemma~\ref{lm:eqpres1v}, $\alpha_q(pc \land x \eq y) = \alpha_q(\alpha_q(pc) \land x \eq y)$.
\item case $s = \textbf{assume}(\pv{x} \neq \pv{y})$, Since $q' = q$, we need to
  show that $\alpha_{q}(pc \land x \deq y) = \alpha_{q}(\alpha_q(pc)
  \land x \deq y)$. By Lemma~\ref{lm:idem}, $\alpha_q(pc) = \alpha_q(\alpha_q(pc))$. By Lemma~\ref{lm:deqpres1v}, $\alpha_{q}(pc \land x \deq y) = \alpha_{q}(\alpha_q(pc)
  \land x \deq y)$.

\item case $s = \pv{x} \gets \pv{y}$. W.l.o.g., assume $q' = q[\pv{x}\mapsto
  x']$, for some constant $x' \not\in \const(pc)$.
  By Lemma~\ref{lm:idem}, $\alpha_q(pc) = \alpha_q(\alpha_q(pc))$.
  By Lemma~\ref{lm:propequiv} (case~1), $\alpha_{\consts(q)
    \cup \{x'\}}(pc \land x' \eq y) = \alpha_{\consts(q)
    \cup \{x'\}}(\alpha_q(pc) \land x' \eq y)$. By Lemma~\ref{lm:subset}, $\alpha_{q'}(pc \land x' \eq y ) = \alpha_{q'}(\alpha_q(pc) \land x' \eq y)$,
  since $\consts(q') \subseteq (\consts(q) \cup \{x'\})$.
\item case $s = \pv{x} \gets f(\Vec{y})$. W.l.o.g., $q' = q[\pv{x}\mapsto
  x']$ for some constant $x' \not\in \const(pc)$. There are two cases: (a) there
  is a term $t \in \terms(pc)$ s.t. $pc \vdash t \eq f(\Vec{y})$, (b) there is no
  such term $t$.
  \begin{enumerate}[label=(\alph*)]
  \item By the memoizing property of CUP, there is a program variable $\pv{z}$
    s.t. $q(\pv{z}) = z$ and $pc \vdash z \eq f(\Vec{y})$. Therefore, by
    definition of $\alpha_q$, $\alpha_q(pc) \vdash z \eq f(\Vec{y})$. The rest
    of the proof is identical to the case of $s = \pv{x} \gets \pv{z}$.
  \item Since there is no term $t\in\terms(pc)$ s.t. $pc \vdash t \eq f(\Vec{y})$, there is also no such term in $\terms(\alpha_q(pc))$ as well. By Lemma~\ref{lm:idem}, $\alpha_q(pc) = \alpha_q(\alpha_q(pc))$.
    By Lemma~\ref{lm:propequiv} (case~2), $\alpha_{\const(q)\cup \{x'\}}(pc \land x \eq f(\Vec{y})) = \alpha_{\const(q)\cup \{x'\}}(\alpha_q(pc) \land x \eq f(\Vec{y}))$.
    By Lemma~\ref{lm:subset}, $\alpha_{q'}(pc \land x \eq f(\Vec{y}))
    = \alpha_{q'}(\alpha_q(pc) \land x \eq f(\Vec{y}))$ since $\consts(q') \subseteq (\consts(q) \cup \{x'\})$.\hfill\ProofSymbol
    \end{enumerate}
    \end{itemize}
  \end{proof}

\begin{corollary}
\label{cor:main}
  For a CUP $P$, the relation $\rho \eqdef \{(c, \exba(c)) \mid c \in
  \Reach(\TS_P)\}$ is a bisimulation from $\TS_P$ to $\exba(\TS_P)$.
\end{corollary}
Note that for an arbitrary UP,  $\alpha_{b}$ induces a simulation (since $\alpha_b$ only weakens path conditions).

By construction, for any configuration in an abstract system constructed using $\alpha_{b}$, the path condition will be at most depth-1. In \Cref{sec:char}, we use this
property to build a logical characterization of CUP and show that
reachability of CUP programs is decidable.


\section{Logical Characterization of CUP}
\label{sec:char}
In this section, we show that for any CUP program $P$, all reachable configurations of $P$ can be characterized using formulas in EUF, whose size is bounded by the number of program variables in $P$. 

\begin{theorem}[Logical Characterization of CUP]\label{th:lccup}
  For any CUP $P$, there exists an inductive assertion map $\inv$, ranging over EUF formulas of depth at most 1, that characterizes the reachable configurations of $P$.
\end{theorem}

The first step in the proof is to compose the renaming abstraction~(\cref{def:rename-abstraction}) with the base abstraction~(\cref{def:alpha-abstraction}). We denote the composition with $\alpha_{b,\rn}$, i.e., $\alpha_{b,\rn} \eqdef \alpha_b \circ \alpha_\rn$. \Cref{cor:main} and \Cref{thm:renaming-bisimilar} ensures that $\alpha_{b, \rn}$ is sound and
complete for CUP. We split the rest of the proof into two cases:
CUPs restricted to unary functions, called 1-CUP, followed by arbitrary CUPs.

\begin{proof}[\Cref{th:lccup}, 1-CUP]
  Let $\Sigma^1$ be a signature containing function symbols of arity atmost $1$,
  $\Sigma^1 \eqdef (\consts, \funs^1, \{\eq, \deq\})$. Let $\Gamma$ be a set of
  literals in $\Sigma^1$ and $V$ be a set of constants. By the definition of
  $V$-base abstraction~(\Cref{def:vabst}), $\alpha_V(\Gamma) = \beta_\eq \land
  \beta_\deq \land \beta_\funs$. $\beta_\eq$ and $\beta_\deq$ are over constants
  in $V$. $\beta_\funs$ contains two types of literals: $\beta_{\funs_V}$ and
  $\beta_{\funs_W}$. $\beta_{\funs_V}$ are 1 depth literals over constants in
  $V$. $\beta_{\funs_W}$ are literals of the form $v \eq f(\Vec{w})$ where $v
  \in V$ and $\Vec{w}$ is a list of constants, at least one of which is in $V$:
  $\Vec{w}\cap V \neq \emptyset$ and $\Vec{w}\not\subseteq V$. Since $\Gamma$
  can only have unary functions, $\beta_{\funs_W} = \emptyset$. Therefore, all
  literals in $\alpha_V(\Gamma)$ are of depth at most 1 and only contain
  constants from $V$. Hence, there are only finitely many configurations in $\alpha_{b,
    \rn}(\TS_P)$.
  Therefore,
\[
  \inv(s) \eqdef \bigvee \{ pc \mid \langle s, q_0, pc \rangle \in
  \Reach(\alpha_{b,\rn}(\TS_P))\}
\]
is an inductive assertion map, ranging over formulas for depth at most 1, that characterizes the reachable configurations of
$P$. Moreover, the size
of each disjunct in $\inv(s)$ is polynomial in the number of program variables
and functions in $P$.
\end{proof}
An interesting consequence of the above proof is that, for
1-CUPs, $\alpha_b$ is efficiently computable (since, $\beta_{\funs_W} = \emptyset$). Thus, the transition system $\alpha_{b,\rn}(\mathcal{S}_P)$ is finite, and can be constructed on-the-fly. Hence, reachability of $1$-CUP is in PSPACE.


\begin{proof}[\Cref{th:lccup}, general case]
  In general, CUP programs can contain unary and non-unary functions. Therefore, the $V$-base abstraction~(\Cref{def:vabst}) may introduce fresh constants. We use the cover
  abstraction~(\Cref{def:cover-abs}) to eliminate these fresh constants. By
  \Cref{thm:cover-bisimilar}, $\alpha_\cover(\alpha_{b, \rn}(\TS_P))$ is
  bisimilar to $\alpha_{b, \rn}(\TS_P)$.
  Notice that all the fresh constants introduced by the $V$-base abstraction are arguments to function
  applications. Therefore, all consequences of eliminating the fresh constants
  are Horn clauses of the form $\bigwedge_i (x_i \eq y_i)\limp x \eq y$, where $x_i, y_i, x, y
  \in \const_0$. Since $V$-basis is of depth at most 1, cover of the $V$-basis 
  is also of depth at most 1. Since there are only finitely many formulas of depth at most~1 over
  $\consts_0$, $\alpha_\cover(\alpha_{b,
    \rn}(\TS_P))$ has only finitely many configurations. Hence,
\[
  \inv(s) \eqdef \bigvee \{ pc \mid \langle s, q_0, pc \rangle \in
  \Reach(\alpha_\cover(\alpha_{b,\rn}(\TS_P))\}
\]
is an inductive assertion map that characterizes the reachable configurations of
$P$ and ranges over depth-1 formulas.
\end{proof}
Consider the CUP shown in \Cref{fig:cup}. At line~9, the $\alpha_{b,\rn}$
abstraction produces the following abstract $pc$: $x_0 \eq f(a_0, w) \land y_0
\eq f(b_0, w) \land c_0 \eq d_0$. Using cover to eliminate the constant $w$
gives us $\cover w \cdot pc = (a_0 \eq b_0 \limp x_0 \eq y_0) \land c_0 \eq
d_0$, which is exactly the invariant assertion mapping $\inv(9)$ at line~9.

We have seen that all CUP programs have an inductive assertion map that  characterizes their reachable configurations and ranges over a finite set of formulas. Therefore,
\begin{corollary}\label{cor:cupdec}
  CUP reachability is decidable.
\end{corollary}
\subsection{Relationship to \cite{DBLP:journals/pacmpl/MathurMV19}} In~\cite{DBLP:journals/pacmpl/MathurMV19}, \Cref{cor:cupdec} is proven by constructing a deterministic finite automaton that accepts all \emph{feasible} coherent executions.\footnote{
In our setting, feasible coherent executions correspond to paths in the transition system of any CUP.} 
However, the construction fails for the executions of the CUP in \Cref{fig:cup}: the  execution that reaches a terminal configuration is infeasible, but it is (wrongfully) accepted by the automaton. Intuitively, the reason is that the automaton is deterministic and its states are not sufficiently expressive. The states of the automaton keep track of equalities between program variables (which correspond to $\beta_\eq$ in our abstraction), disequalities between them ($\beta_\deq$ in our case), and partial function interpretations ($\beta_\funs$). However, the partial function interpretations are restricted to $\beta_{\funs_V}$, i.e., do not allow auxiliary constants that are not assigned to program variables. Thus, they are unable to keep track of $x_0 \eq f(a_0, w) \land y_0
\eq f(b_0, w) \land c_0 \eq d_0$ in line 9, which is essential for showing infeasibility of the execution. Eliminating the auxiliary constants, as we do in the cover abstraction, does not remedy the situation since it introduces a disjunction $(a_0 \deq b_0 \wedge c_0 \eq d_0) \lor (x_0 \eq y_0 \wedge c_0 \eq d_0)$, which 
the deterministic automaton does not capture. 

\subsection{Computing a Finite Abstraction}
\newcommand{\normal}{n}
\newcommand{\normalW}{nw}
We have shown that CUP programs are bisimilar to finite state systems. However, all our proofs depend on $\exba$, which was not assumed to be computable. In this section, we show how to implement $\exba$, and, thereby, show how to compute a finite state system that is bisimilar to a CUP program. Note that our prior results are independent of this section.

The main difficulty is in naming the fresh constants, which we always refer to as $W$, that are introduced by the base abstraction. Since we require that base abstraction is canonical, the naming has to be unique. Furthermore, we have to show that the number of such $W$ constants is bounded. We solve both of these problems by proposing a deterministic naming scheme. The scheme is determined by a normalization function $\normal_V$  that replaces all the fresh constants in a $V$-basis with canonical constants.


Let $\beta$ be a $V$-basis. We denote the auxiliary constants in $\beta$ ($\consts(\beta) \setminus V$) by $W = \{w_0, w_1, \ldots\}$, and by `$\hole$' some unused constant that we call a \emph{hole}. Recall that constants from $W$ may only appear in literals of the form $v \eq f(\Vec{w})$.
%
We define the set of $W$-templates as the set of all terms $f(\Vec{a})$, where each element in $\Vec{a}$ is either a hole or a constant in $W$. A term $t$ \emph{matches} a template $f(\Vec{a})$ if $t = f(\Vec{b})$, and $\Vec{a}$ and $\Vec{b}$ agree on all constants in $W$. For example, let $\xi$ be the template $f(\hole, w_1, \hole, w_2)$. The term $f(a, w_1, b, w_2)$ matches $\xi$, but $f(w_0, w_1, b, w_2)$ does not, because one of the holes is filled with $w_0 \in W$. We say that a literal $v \eq f(\Vec{b})$ matches a template $\xi$ if $f(\Vec{b})$ matches $\xi$. The $W$-context of a $W$-template $\xi$ in a set of literals $L$, denoted $\context_L(\xi)$, is the set $\context_L(\xi) \eqdef \{\ell[W\mapsto \hole] \mid \ell\in L \land \ell \text{ matches } \xi\}$, 
where $\ell[W\mapsto \hole]$ means that all occurrences of constants in $W$ are replaced with a hole. For example, let $\xi = f(\hole, w_1, w_2, \hole)$ and $L = \{v \eq f(a, w_1, w_2, b), u \eq f(c, w_1, w_2, a), w \eq f(x, w_1, w_2, b), x \eq g(x, w_1, w_2, b))\}$ then $\context_L(\xi) =\{v \eq f(a, \hole, \hole, b), u \eq f(c, \hole, \hole, a), w \eq f(x, \hole, \hole, b)\}$.

Since $V$ and $\funs$ are finite, the number of $W$-contexts is finite, independent of $W$. Let $w_Z$ be a fresh constant for context $Z$. 

\begin{definition}[Normalization Function]
\label{def:normfunc}
The normalization function $\normal_{V}(\beta)$ is defined as follows:
\begin{enumerate}[(1)]
    \item for each $t \in \terms(\Gamma)$ s.t. $\const(t)\cap W \neq \emptyset$, create a template $\xi$ by dropping all constants not in $W$. Let $\Xi$ denote the set of templates so obtained.
    \item Let $\mathit{Ctx} \eqdef \{\context_{\Gamma}(\xi) \mid \xi \in \Xi\}$.
    \item For each $\ell \in \Gamma$, if $\ell[W\mapsto \hole] \in Z$ for some $Z \in \mathit{Ctx}$, then replace all occurrences of $W$ in $\ell$ with $w_Z$.
\end{enumerate}
\end{definition}

The normalization preserves $V$-equivalence of $\beta$ because it renames local constants, while maintaining all consequences that are derivable through them. That is, $\normal_V(\beta) \equiv_V \beta$. Furthermore, $\normal_V(\beta)$ is cannonical.

Therefore, given a set of literals $\Gamma$, we use $\normal_V(\beta)$ as a computable implementation of the $V$-base abstraction, $\alpha_V$~(\cref{def:vabst}). That is, $\alpha_V(\Gamma) \eqdef \normal_V(\beta)$ where $\langle W, \beta, \delta\rangle \in \base(\Gamma, V)$. Even though $\normal_V(\beta)$ may not be a part of a $V$-basis for $\Gamma$, 
it satisfies all the properties used in the proof of \Cref{thm:main}.

%

We define the normalizing abstraction in the usual way:
\begin{definition}[Normalizing abstraction]
The normalizing abstraction function $\alpha_\normal: \Conf \to \Conf$ is defined by
\[
\alpha_\normal(\langle s, \qinit, pc\rangle) \eqdef \langle s,\qinit,\normal(pc)\rangle
\]
\end{definition}

Let $\alpha_{b,\rn, \normal} \eqdef \alpha_{b} \circ \alpha_\rn \circ \alpha_\normal$ be the composition of normalization abstraction with renaming and base abstraction where $\alpha_b$ is implemented using normalization. Notice that, for any state $c = \langle s, q, pc\rangle$, $\alpha_{b, \rn, \normal}(c)$ is computed by first computing \emph{any} $V$-basis of $pc$, applying $\normal_q$, renaming all $\const(q)$ constants to $\qinit$, and applying $\normal_\qinit$. The second normalization is required to ensure that the fresh constants are canonical with respect to $\qinit$. By definition $\alpha_{b, \rn, n}$ is computable. Hence, it can be used to compute the finite abstraction of any CUP.

\begin{theorem}
\label{thm:ttt}
For a CUP $P$, the finite abstract transition system $\alpha_{b', \rn, \normal}(\TS_P)$ is bisimilar to $P$ and is computable.
\end{theorem}

\Cref{thm:ttt} implies that any property that is decidable over a finite transition system is also decidable over CUPs. In particular, temporal logic model checking is decidable.


\section{Conclusion}
\label{sec:conclusion}

In this paper, we study theoretical properties of Coherent Uninterpreted Programs (CUPs) that have been recently proposed by Mathur et al.~\cite{DBLP:journals/pacmpl/MathurMV19}. We identify a bug in the original paper, and provide an alternative proof of decidability of the reachability problem for CUP. More significantly, we provide a logical characterization of CUP. First, we show that inductive invariant of CUP is describable by shallow formulas. Hence, the set of all candidate invariants can be effectively enumerated. Second, we show that CUPs are bisimilar to finite transition systems. Thus, while they are formally infinite state, they are not any more expressive than a finite state system. Third, we propose an algorithm to compute a finite transition system of a CUP. This lifts all existing results on finite state model checking to CUPs. 

In the paper, we have focused on the core result of Mathur et al, and have left out several interesting extensions. In~\cite{DBLP:journals/pacmpl/MathurMV19}, the notion of CUP is extended with $k$-coherence -- a UP $P$ is $k$-coherent if it is possible to transform $P$ into a CUP $\hat{P}$ by adding $k$ \emph{ghost} variables to $P$. This is an interesting extension since it makes potentially many more programs amenable to decidable verification. We observe that addition of \emph{ghost} variables is a form of abstraction. Thus, invariants of $\hat{P}$ can be translated to invariants of $P$ using techniques of Namjoshi et al.~\cite{DBLP:conf/sas/NamjoshiZ13,DBLP:conf/vmcai/Namjoshi03}. This essentially amounts to existentially eliminating ghost variables from the invariant of $\hat{P}$. Such elimination increases the depth of terms in the invariant at most by one for each variable eliminated. Thus, we conjecture that $k$-coherent programs are characterized by invariants with terms of depth at most $k$.

Mathur et al.~\cite{DBLP:journals/pacmpl/MathurMV19} extend their results to recursive UP programs (i.e., UP programs with recursive procedures). We believe our logical characterization results extend to this setting as well. In this case, both the invariants and procedure summaries (i.e., procedure pre- and post-conditions) are described using terms of depth at most 1. 

Our results also hold when CUPs are extended with  simple axiom schemes, as in~\cite{DBLP:conf/tacas/MathurM020}, while for most non-trivial axiom schemes CUPs become undecidable.

Perhaps most interestingly, our results suggest efficient verification algorithms for CUPs and 
interesting abstraction for UPs. Since the space of invariant candidates is finite, it can be enumerated, for example, using implicit predicate abstraction. For CUPs, this is a complete verification method. For UPs it is an abstraction. Most importantly, it does not require prior knowledge to whether an UP is a CUP! 



\subsubsection*{Acknowledgment}
The research leading to these results has received funding from the
European Research Council under the European Union's Horizon 2020 research and
innovation programme (grant agreement No [759102-SVIS]).
This research was partially supported by the United States-Israel Binational Science Foundation (BSF) grant No. 2016260, and the Israeli Science Foundation (ISF) grant No. 1810/18. We also acknowledge the support of the Natural Sciences and Engineering Research Council of Canada (NSERC).

\bibliographystyle{IEEEtranS}
\bibliography{ref}
\clearpage
\newpage
\appendices
\section{Additional Background on EUF}
\label{sec:euf_extra}

\begin{figure}[t]
  \begin{mathpar}
    \prftree[r]{\textsc{Refl}}{\Gamma \vdash x \eq x}\qquad
    \prftree[r]{\textsc{Symm}}{\Gamma \vdash x \eq y}{\Gamma \vdash y \eq x}\\
    \prftree[r]{\textsc{Trans}}{\Gamma \vdash x \eq y}{\Gamma \vdash y \eq
      z}{\Gamma \vdash x \eq z}\\
    \prftree[r]{\textsc{Cong}}{\Gamma \vdash x_1 \eq y_1 \quad\cdots\quad \Gamma \vdash x_n
      \eq y_n}{\Gamma \vdash f(x_1, \ldots, x_n) \eq f(y_1, \ldots, y_n)}\\
    \prftree[r]{\textsc{EqNeq}}{\Gamma\vdash x \eq y}{x \deq y \in
      \Gamma}{\Gamma\vdash \bot}
    \quad
    \prftree[r]{\textsc{PMod}}{\Gamma \vdash \ell}{\Gamma \vdash x \eq
      y}{\Gamma \vdash \ell[x \mapsto y]}
  \end{mathpar}
  \caption{Proof system $\peuf$.}
\label{fig:peuf}
\end{figure}

In this section, we formalize some of the concepts about EUF that are well known
and have been excluded from the main content of the paper due to space
limitations.

The proof rules of the proof system $\peuf$ for EUF shown in \Cref{fig:peuf}.
These are the usual rules. The exception is \textsc{PMod} that is a form of
paramodulation. It is used to derive new literals by substituing equal for
equal. While not typically included in the proof rules for EUF, \textsc{PMod} is
used implicitly in the congruence graph algorithms, and in interpolation over
EUF.

A deductive (\euf) closure, $\Gamma^*$, of a set of literals $\Gamma$ is defined
as: $\Gamma^* \eqdef \{\ell \mid \Gamma \vdash \ell\}$. A set $\Gamma$ is
deductively closed if $\Gamma = \Gamma^*$.

For a satisfiable set $\Gamma$ of EUF literals, and $a, b \in \terms(\Gamma)$:
\begin{enumerate}[(1)]
\item $\Gamma \models a \eq b$ iff $a \eq b \in \Gamma^*$
\item $\Gamma \models a \deq b$ iff $\bot \in (\Gamma \cup \{a \eq b\})^*$
\end{enumerate}

Note that $\Gamma \models a \deq b$ does not imply $\Gamma \vdash a \deq b$,
since $\peuf$ has no \textsc{Hyp} and \textsc{Contra} proof rules.

Depth of a term is formally defined as follows:
\[\depth(t) = \begin{cases}
    0  & \text{if $t \in \consts$}\\
    1 + \max_i(\depth(t_i)) & \text{if $t = f(t_0, \ldots, t_k)$}
  \end{cases}
\]

\section{Proofs}
\label{sec:proofs}
Given a set of literals $\Gamma$ and a set of constants $V$, let $\bbase(\Gamma,
V) \eqdef \{\beta \mid \exists W, \delta \cdot \langle W, \beta, \delta\rangle
\in \base(\Gamma, V)\}$.
\begin{alemma}\label{lm:alpheq}\hg{new}
  Let $\varphi_1$ and $\varphi_2$ be two sets of literals and $V$ be a set of
  constants. Then, the following three statements are equivalent:
  \begin{enumerate}[(1)]
  \item $\alpha_V(\varphi_1) = \alpha_V(\varphi_2)$
  \item $\bbase(\varphi_1, V) \cap \bbase(\varphi_2, V) \neq \emptyset$
    \item $\bbase(\varphi_1, V) = \bbase(\varphi_2, V)$
    \end{enumerate}
\end{alemma}
\begin{alemma}\label{lm:trmingam}
  Let $\Gamma$ be a set of literals, $v \in \const(\Gamma)$. If $\Gamma \vdash v \eq f(t_1,\ldots,t_n)$ for some term $f(t_1,\ldots,t_n) \in \terms(\Sigma)$ then there exists a term $f(t'_1,\ldots,t'_n) \in \terms(\Gamma)$ s.t. $\Gamma \vdash v \eq f(t'_1,\ldots,t'_n)$ and $\Gamma \vdash t_i \eq t_i'$ for all $1 \leq i \leq n$.
\end{alemma}
\begin{alemma}\label{lm:trmingamdeq}
  Let $\Gamma$ be a set of literals, $v \in \const(\Gamma)$. If $\Gamma \vdash v \deq f(t_1,\ldots,t_n)$ for some term $f(t_1,\ldots,t_n) \in \terms(\Sigma)$ then there exists a term $f(t'_1,\ldots,t'_n) \in \terms(\Gamma)$ s.t. $\Gamma \vdash v \deq f(t'_1,\ldots,t'_n)$ and $\Gamma \vdash t_i \eq t_i'$ for all $1 \leq i \leq n$.
\end{alemma}
\setcounter{lemma}{1}
\begin{lemma}
  Let $\Gamma$ be a set of literals,  $x$ and $y$ be two constants in
  $\consts(\Gamma)$, $V \subseteq \consts(\Gamma)$ be
  a purifier for $\{x, y\} \subseteq V$, and $\beta =
  \alpha_V(\Gamma)$. Then, for any $u, v \in V$
\[
  (\Gamma \land x \eq y \vdash u \eq v)
  \iff
 ( \beta \land x \eq y \vdash u \eq v)
\]
\end{lemma}
\begin{proof}
  By the definition of $\beta$, $(\Gamma \vdash u \eq v)
  \iff (\beta \vdash u \eq v)$. Thus, assume that $\Gamma \not \vdash u \eq v$.

  The only-if direction is trivial since $\beta$ is weaker than $\Gamma$.

  For the if-direction, By Lemma~\ref{lm:sprtrm}, there are superterms $s_1[x]$
  and $s_2[y]$ of $x$ and $y$, respectively, s.t. $\Gamma \vdash \{u \eq s_1[x],
  v \eq s_2[y]\}$, and $(\Gamma \land x \eq y) \vdash (s_1[x] \eq s_2[y])$. The
  proof proceeds by induction on the maximum depth of $s_1$ and $s_2$. The base
  case, $s_1 = x$ and $s_2 = y$, is trivial.

  For the inductive case, we show one sub-cases, others are similar. Assume
  $s_1 = f(t_1[x], \vr)$ and $s_2 = f(t_2[y], \vr)$, for some terms $t_1[x]$,
  $t_2[y]$, $\vr$, and a function $f$. Furthermore, $(\Gamma \land x \eq
  y) \vdash t_1[x] \eq t_2[y]$. Since $\Gamma \vdash \{u \eq f(t_1[x], \vr), v \eq f(t_2[y], \vr)\}$, by Lemma~\ref{lm:trmingam}, there exists terms $f(t_1', \vr_1), f(t_2', \vr_2) \in \terms(\Gamma)$ s.t. $\Gamma \vdash \{u \eq f(t_1', \vr_1), t_1[x] \eq t_1', \vr \eq \vr_1, v \eq f(t_2', \vr_2), t_2[y] \eq t_2', \vr \eq \vr_2 \}$.

  Since $V$ is a purifier for $\{x, y\}$, there are $x', y' \in V$ s.t.
  $\Gamma \vdash \{x' \eq t_1', y' \eq t_2'\}$, and $\Gamma \vdash \{u \eq
  f(x', \vr), v \eq f(y', \vr)\}$. By construction, $\beta \vdash \{u \eq
  f(x', \vw), v \eq f(y', \vw)\}$, for some constants $\vw \in \consts(\beta)$. By IH,
  $(\beta \land x \eq y) \vdash x' \eq y'$. Hence, by congruence, $(\beta \land
  x \eq y) \vdash v \eq u$.
  \end{proof}
\begin{lemma}
  Let $\Gamma$ be a set of literals, $x$ and $y$ be two constants in
  $\consts(\Gamma)$, $V \subseteq \consts(\Gamma)$ be a purifier for $\{x, y\}
  \subseteq V$, and $\beta = \alpha_V(\Gamma)$. Then, for any $u, v \in V$
  \[
    (\Gamma \land x \eq y \vdash u \deq v)
    \iff
    ( \beta \land x \eq y \vdash u \deq v)
  \]
\end{lemma}
\begin{proof}
  By the definition of $\beta$, $(\Gamma \vdash u \deq v) \iff (\beta \vdash u \deq
  v)$. Assume $\Gamma \not \vdash u \deq v$. Then, there is a term $t \in
  \terms(\Sigma)$, s.t. $\Gamma \vdash u \deq t$ and $(\Gamma \land x \eq y)
  \vdash v \eq t$. By Lemma~\ref{lm:sprtrm}, $\Gamma \vdash t \eq s[y]$. We case split on whether $s[y]$ is $y$ itself or some superterm of $y$.
 \begin{itemize}
     \item case $s[y] = s$, Since $\Gamma \vdash t \eq y$, $(\Gamma \land x \eq y) \vdash v \eq y$ and $\Gamma \vdash u \deq y$. By Lemma~\ref{lm:Visenough}, $(\beta \land x \eq y) \vdash v \eq y$. By definition $\beta \vdash u \deq y$. Therefore, $(\beta \land x \eq y) \vdash u \deq v$.

     \item case $s[y] = f(t_1,\ldots, t_n)$, where at least one $t_i$ is a superterm of $y$. Since $\Gamma \vdash t \eq f(t_1,\ldots, t_n)$, $\Gamma \vdash u \deq f(t_1,\ldots, t_n)$. By Lemma~\ref{lm:trmingamdeq}, there exists a term $f(t_1',\ldots, t_n') \in \terms(\Gamma)$ s.t., $\Gamma \vdash \{ u \deq f(t_1',\ldots, t_n'), t_1' \eq t_1,\ldots, t_n' \eq t_n\}$. Since $f(t_1',\ldots, t_n') \in \terms(\Gamma)$, $\Gamma \vdash f(t_1',\ldots, t_n') \eq s[y]$, and $V$ is a purifier for $y$ in $\Gamma$, there exists a constant $y' \in V$ s.t. $\Gamma \vdash y' \eq f(t_1',\ldots, t_n')$. Therefore, $(\Gamma \land x \eq y) \vdash v \eq y'$ and $\Gamma \vdash u \deq y'$. By Lemma~\ref{lm:Visenough}, $(\beta \land x \eq y) \vdash v \eq y'$. By the definition of $\beta$, $\beta \vdash u \deq y'$. Therefore, $(\beta \land x \eq y) \vdash u \deq v$.

 \end{itemize} 
\end{proof}
\setcounter{lemma}{5}
\begin{lemma}
Let $V$ be a set of constants, $\varphi_1$ and $\varphi_2$ be two sets of literals s.t. $\alpha_V(\varphi_1) = \alpha_V(\varphi_2)$, and $V$ be a purifier for $\{x, y\}$ in both $\varphi_1$ and $\varphi_2$. Then, $\alpha_V(\varphi_1 \land x \eq y) = \alpha_V(\varphi_2 \land x \eq y)$
\end{lemma}
\begin{proof}
  Let $\beta \in \bbase(\varphi, V)$. Let $\beta' \in \bbase(\varphi_1 \land x
  \eq y, V)$ s.t. $\beta \subseteq \beta'$. Let $L_{\eq}$ be a set of equalities
  between constants in $V$, $L_{\deq}$ be a set of disequalities between
  constants in $V$, and $L_{\funs}$ is a set of equalities of the form $v \eq
  f(\Vec{w})$ where $v \in V$, and $\Vec{w}$ is a set of constants, some of
  which are in $V$, and the rest are not in
  $\const(\varphi_1)\cup\const(\varphi_2)$. Let $\beta' = \beta \cup L_{\eq}\cup
  L_{\funs} \cup L_{\deq}$.

  By Lemma~\ref{lm:Visenough} and Lemma~\ref{lm:Visenoughdeq}, $\beta \land x
  \eq y \vdash \ell$ for all $\ell \in L_{\eq}\cup L_{\deq}$.

  Next, we prove that for all $\ell \in L_\funs$, $\beta \land x \eq y \vdash
  \ell$. In the following, assume that $u, v \in V$ and $w \not \in
  \const(\Gamma)\cup V$. We assume that $\ell = v \eq f(u, w)$. All other cases
  are similar. We have $\beta' \vdash v \eq f(u, w)$ iff $(\varphi_1 \land x \eq
  y) \vdash v \eq f(u, t)$ for some term $t \in \terms(\Sigma)$ and there no $v'
  \in V$ s.t. $(\varphi_1 \land x \eq y) \vdash v' \eq t$. If $\varphi_1 \vdash
  v \eq f(u, t)$ then $\beta \vdash v \eq f(u, w)$ by definition. Assume that
  $\varphi_1 \nvdash v \eq f(u, t)$. By Lemma~\ref{lm:sprtrm}, we have
  $\varphi_1 \vdash v \eq s_1[x]$ and $\varphi_1 \vdash f(u, t) \eq s_2[y]$ and
  $\varphi_1 \wedge x \eq y \vdash s_1[x] \eq s_2[y]$. We case split on
  $s_2[y]$.
\begin{enumerate}
\item case $s_2[y] = y$. We have $\varphi_1 \vdash f(u, t) \eq y$. From
  $(\varphi_1 \land x \eq y) \vdash v \eq f(u, t)$ and $\varphi_1 \vdash f(u, t)
  \eq y$, we have $(\varphi_1 \land x \eq y) \vdash v \eq y$. From
  \Cref{lm:Visenough}, we have $(\beta \land x \eq y) \vdash v \eq y$. Since
  $\varphi_1 \vdash f(u, t) \eq y$, $\beta \vdash f(u, w) \eq y$ by definition.
  Hence, $(\beta \land x \eq y) \vdash v \eq f(u, w)$.
\item case $s_2[y] = g(b_1,\ldots,b_n)$, where, $b_i$ is a superterm of $y$ for
  at least one $i$. We have, $\varphi_1 \vdash f(u, t) \eq g(b_1,\ldots,b_n)$.
  It is either the case that there exist a term $t'$ s.t. $g(b_1,\ldots, b_n)
  \eq t' \in \varphi_1$ and $\varphi_1 \vdash t' \eq f(u, t)$, or $g = f$ and
  $\varphi_1 \vdash\{u \eq b_1, t \eq b_2\}$.
        \begin{enumerate}
        \item case there exist a term $t'$ s.t. $g(b_1,\ldots, b_n) \eq t' \in
          \varphi_1$ and $\varphi_1 \vdash t' \eq f(u, t)$. Since $g(b_1,\ldots,
          b_n) \in \terms(\varphi_1)$ and $V$ is a purifier for $y$ in
          $\varphi_1$, there exists a $y'\in V$ s.t. $\varphi_1 \vdash y' \eq
          g(b_1,\ldots,b_n)$. Therefore, $\varphi_1\vdash y' \eq f(u, t)$ and
          $(\varphi_1 \land x \eq y) \vdash v \eq y'$. By definition, $\beta
          \vdash y' \eq f(u, w)$. By \Cref{lm:Visenough}, we have $(\beta \land
          x \eq y) \vdash v \eq y'$. Hence, $(\beta \land x \eq y) \vdash v \eq
          f(u, w)$.
        \item case $g = f$ and $\varphi_1 \vdash\{u \eq b_1, t \eq b_2\}$.
          It has to be the case that $s_1[x] = f(a_1, a_2)$. We have $\varphi_1 \vdash v \eq f(a_1, a_2)$ where $a_1$ or
          $a_2$ is a superterm of $x$. By \Cref{lm:trmingam}, we have a
          term $f(a_1', a_2') \in \terms(\varphi_1)$ s.t. $\varphi_1 \vdash \{v
          \eq f(a_1', a_2'), a_1' \eq a_1, a_2' \eq a_2\}$. We case split on
          whether $a_1$ or $a_2$ is a superterm of $x$:
            \begin{enumerate}
            \item case $a_1$ is a superterm of $x$. We have, $(\varphi_1 \land x
              \eq y) \vdash a_1 \eq b_1$. Since $\varphi_1 \vdash a_1' \eq a_1$,
              $a_1' \in \terms(\varphi_1)$, and $V$ is a purifier for $x$ in
              $\varphi_1$, there must exists a constant $x' \in V$ s.t.
              $\varphi_1 \vdash x' \eq a_1'$. Since $(\varphi_1 \land x \eq y)
              \vdash a_1 \eq b_1$, $(\varphi_1 \land x \eq y) \vdash x' \eq
              b_1$. From $\varphi_1 \vdash u \eq b_1$, we have $\varphi_1 \land
              x \eq y \vdash x' \eq u$. By \Cref{lm:Visenough}, we have $\beta
              \land x \eq y \vdash x' \eq u$. Since $\varphi \wedge v \eq
              f(a_1', a_2')$ and $\varphi \wedge x' \eq a_1'$, we have
              $\varphi_1 \vdash v \eq f(x', a_2')$ and hence $\beta \vdash v \eq
              f(x', w)$ be definition. Since $\beta \vdash v \eq f(x', w)$ and
              $\beta \land x \eq y \vdash x' \eq u$, $\beta \land
              x \eq y \vdash v \eq f(u, w)$.
            \item case $a_2$ is a superterm of $x$. We have $(\varphi \land x
              \eq y) \vdash a_2 \eq b_2$. Since $\varphi_1 \vdash a_2' \eq a_2$,
              $a_2' \in \terms(\varphi_1)$, and $V$ is a purifier for $x$ in
              $\varphi_1$, there must exists a constant $x' \in V$ s.t.
              $\varphi_1 \vdash x' \eq a_2'$. Since $(\varphi_1 \land x \eq y)
              \vdash a_2 \eq b_2$, $(\varphi_1 \land x \eq y) \vdash x' \eq
              b_2$. However, $\varphi \vdash b_2 \eq t$ and hence $(\varphi_1
              \land x \eq y) \vdash t \eq x'$ which contradicts our assumption
              that there is no $v' \in V$ such that $(\varphi_1 \land x \eq y)
              \vdash v' \eq t$ .
            \end{enumerate}
        \end{enumerate}
\end{enumerate}

Since $\alpha_V(\varphi_1) = \alpha_V(\varphi_2)$, by \Cref{lm:alpheq} we have
$\beta \in \bbase(\varphi_2, V)$. Therefore, $\beta' \in \bbase(\varphi_2 \land
x \eq y, V)$ as well. Since $\bbase(\varphi_1 \land x \eq y, V) \cap
\bbase(\varphi_2 \land x \eq y, V) \neq \emptyset$, by \Cref{lm:alpheq},
$\alpha_V(\varphi_1\land x \eq y) = \alpha_V(\varphi_2 \land x \eq y)$.
\end{proof}
\setcounter{lemma}{6}
\begin{lemma}
  Let $V$ be a set of constants, $\varphi_1$, $\varphi_2$ be two sets of
  literals s.t. $\alpha_V(\varphi_1) = \alpha_V(\varphi_2)$. Then, for any $U
  \subseteq V$, $\alpha_U(\varphi_1) = \alpha_U(\varphi_2)$.
\end{lemma}
\begin{proof}
  Follows from $\alpha_U(\alpha_V(\varphi_1)) = \alpha_U(\alpha_V(\varphi_2))$,
  and $\alpha_U(\alpha_V(\varphi_i)) = \alpha_U(\varphi_i)$, for $i \in \{1,
  2\}$.
\end{proof}
\setcounter{lemma}{7}
\begin{lemma}
  Let $V$ be a set of constants s.t. $x, y \in V$, $\varphi_1$ and $\varphi_2$
  be two sets of literals s.t. $\alpha_V(\varphi_1) = \alpha_V(\varphi_2)$.
  Then, $\alpha_V(\varphi_1 \land x \deq y) = \alpha_V(\varphi_2 \land x \deq
  y)$
\end{lemma}
\begin{proof}
  Let $\beta \in \bbase(\varphi_1, V)$. Then, $\beta \land L \in
  \bbase(\varphi_1\land x \deq y, V)$, where, $L = \{x \deq u \mid y \eq u \in
  \beta, u \in V\}$. Since $\alpha_V(\varphi_1)\eq\alpha_V(\varphi_2)$,
  by~\Cref{lm:alpheq}, $\beta \in \bbase(\varphi_2, V)$. Therefore, $\beta \land
  L \in \bbase(\varphi_2\land x\deq y, V)$. Since, $\bbase(\varphi_1\land x \deq
  y, V)\cap \bbase(\varphi_1\land x \deq y, V)\neq \emptyset$, by
  \Cref{lm:alpheq}, $\alpha_V(\varphi_1 \land x \deq y) = \alpha_V(\varphi_2
  \land x \deq y)$.
\end{proof}
\setcounter{lemma}{8}
\begin{lemma}\hg{modified statement}
  Let $V$ be a set of constants, $\varphi_1$, $\varphi_2$ be two sets of
  literals s.t. $\alpha_V(\varphi_1) = \alpha_V(\varphi_2)$, $y \in V,
  \Vec{y}\subseteq V$, $x'$ be a constant s.t. $x' \not \in \const(\varphi_1)
  \cup \const(\varphi_2)$, $V' = V \cup\{x'\}$, and $f(\Vec{y})$ be a term s.t
  there does not exists a term $t \in \terms(\varphi_1) \cup \terms(\varphi_2)$
  s.t. $\varphi_1 \vdash t \eq f(\Vec{y})$ or $\varphi_2 \vdash t \eq
  f(\Vec{y})$. Then,
\begin{enumerate}[(1)]
    \item $\alpha_{V'}(\varphi_1 \land x' \eq y) = \alpha_{V'}(\varphi_2 \land x' \eq y)$
    \item $\alpha_{V'}(\varphi_1 \land x' \eq f(\Vec{y})) = \alpha_{V'}(\varphi_2 \land x' \eq f(\Vec{y}))$
\end{enumerate}
\end{lemma}
\begin{proof}\hg{new}
  \begin{enumerate}[(1)]
  \item Let $\beta \in \bbase(\varphi_1, V)$. By definition of basis, $\beta
    \land L \in \bbase(\varphi_1 \land x' \eq y, V')$, where $L = \{\ell \mid
    \ell[x'\mapsto y] \in \beta\}$. Since $\alpha_V(\varphi_1) =
    \alpha_V(\varphi_2)$, by \Cref{lm:alpheq}, $\beta \in \bbase(\varphi_2, V)$.
    Therefore, $\beta \land L \in \bbase(\varphi_2 \land x' \eq y, V')$ as well.
    Hence $\bbase(\varphi_1 \land x' \eq y, V')\cap \bbase(\varphi_2 \land x'
    \eq y, V') \neq \emptyset$. By \Cref{lm:alpheq}, $\alpha_{V'}(\varphi_1
    \land x' \eq y) = \alpha_{V'}(\varphi_2 \land x' \eq y)$.

  \item Let $\beta \in \bbase(\varphi_1, V)$. By definition of basis, $\beta
    \land X_{def} \in \bbase(\varphi_1 \land x' \eq f(\Vec{y}), V')$, where $X_{def}$ is
    \begin{multline*}
      \{x' \bowtie t \mid \beta \vdash f(\Vec{y}) \bowtie t, \depth(t) = 1,
      \const(t) \cap \const(\beta) \subseteq V\}
    \end{multline*}
    Since $\alpha_V(\varphi_1) = \alpha_V(\varphi_2)$, by \Cref{lm:alpheq},
    $\beta \in \bbase(\varphi_2, V)$. Therefore, $\beta \land X_{def} \in
    \bbase(\varphi_2 \land x' \eq f(\Vec{y}), V')$ as well. Hence
    $\bbase(\varphi_1 \land x' \eq f(\Vec{y}), V')\cap \bbase(\varphi_2 \land x'
    \eq f(\Vec{y}), V') \neq \emptyset$. By \Cref{lm:alpheq},
    $\alpha_{V'}(\varphi_1 \land x' \eq f(\Vec{y})) = \alpha_{V'}(\varphi_2
    \land x' \eq f(\Vec{y}))$.
  \end{enumerate}
\end{proof}
\end{document}